\documentclass[a4paper, 11pt]{article}
\usepackage{authblk}
\usepackage[utf8]{inputenc} 
\usepackage[T1]{fontenc}
\usepackage{lmodern}
\usepackage{graphicx}
\usepackage[english]{babel}
\usepackage[left=2.5cm,right=2.5cm,top=2.5cm,bottom=2.5cm]{geometry}
\usepackage[]{youngtab}
\usepackage{setspace} 
\usepackage{tikz}
\usetikzlibrary{arrows,calc}
\usetikzlibrary{decorations.markings}
\tikzstyle directed=[postaction={decorate,decoration={markings,
    mark=at position .65 with {\arrow{latex}}}}]
\usepackage{amsfonts}
\usepackage{amsmath}
\usepackage{amssymb}
\usepackage{bbm}
\usepackage{amsthm}
\usepackage{color}
\usepackage{enumitem}
\usepackage{hyperref}
\usepackage{systeme}
\usepackage{mathtools}

\newtheorem{definition}{Definition}[section]
\newtheorem{theorem}[definition]{Theorem}
\newtheorem{proposition}[definition]{Proposition}

\newtheorem{remark}[definition]{Remark}
\newtheorem{lemma}[definition]{Lemma}
\newtheorem{RHP}[definition]{Riemann-Hilbert Problem}
\theoremstyle{definition}

\DeclareMathOperator{\Tr}{Tr}
\DeclareMathOperator{\diag}{Diag}
\def\R{\mathbb{R}}
\def\N{\mathbb{N}}

\def\C{\mathbb{C}}
\def\P{\mathbb{P}}
\def\E{\mathbb{E}}
\def\K{\mathbb{K}}

\def\a{\overrightarrow{\mathbf{a}}}
\def\k{\overrightarrow{\mathbf{k}}}

\usepackage{mathtools}
\usepackage{graphicx}
\numberwithin{equation}{section}

\title{A Riemann Hilbert approach to the study of the generating function associated to the Pearcey process}

\author{Thomas Chouteau}
\affil{Univ Angers, CNRS, LAREMA, SFR MATHSTIC, F-49000 Angers, France; \textbf{thomas.chouteau@univ-angers.fr}}

\date{}

\begin{document}
\maketitle
\begin{abstract}
Using Riemann-Hilbert methods, we establish a Tracy-Widom like formula for the generating function of the occupancy numbers of the Pearcey process. This formula is linked to a coupled vector differential equation of order three. We also obtain a non linear coupled heat equation. Combining these two equations we obtain a PDE for the logarithm of the the generating function of the Pearcey process.
\end{abstract}

{\bf Keywords:} Pearcey process, Riemann-Hilbert problems, Integrable operator, Generating function.

\section{Introduction}
\label{sec1}
The Pearcey process is a universal determinantal point process associated to the Pearcey kernel (see \eqref{Kp} for the definition of the kernel). It appeared for the first time in the spectral analysis of random matrices with external sources \cite{BH98}. If $2N$ is the size of matrices, the Pearcey process describes the behaviour, when $N\to\infty$, of eigenvalues near a point where the density of states admits a cusp-like gap. Another model linked with the Pearcey process are $2N$ 1-dimension non intersecting Brownian motions \cite{BK06} starting from 0 at $t=0$ and ending at $\pm a$ (half particles at $a$ and the others at $-a$) at time $t=1$. For this model, in the large $N$-limit, there exists a time $t_c$ such that, for $t<t_c$ the distribution of particles is supported in an interval and for $t>t_c$ it splits in two disjoint intervals. The distribution of particles of the Brownian motion for $t$ close to $t_c$ is described by the Pearcey process. A last example which reveals the universality of the Pearcey process is the one of random skew plane partition \cite{OR05}. Studying the limit shape associated to this model leads to different processes: Beta process, extended Airy process, and extended Pearcey process.\\
As known for the Airy process \cite{TW94,AvM05}, it is possible to express probabilities associated to the Pearcey process in terms of some partial differential equations by studying the Fredholm determinant of the Pearcey kernel operator \eqref{Kp}. Studying the Brownian motion model at several times and deriving the kernel associated to the Pearcey process, Tracy and Widom \cite{TW06}, obtained partial differential equations for the distribution associated to the Pearcey process. Using KP-tau functions and Hirota bilinear equations for the study of random matrices with external sources, Adler and van Moerbeke \cite{AvM07} introduced a non linear fourth order PDE for the Pearcey process. With Riemann-Hilbert methods Bertola and Cafasso \cite{BC12} obtained the same PDE as Adler and van Moerbeke and introduced a new one for the gap probability for the Pearcey process.\\ \\
The aim of this article is to present a Tracy-Widom formula for the generating function of the Pearcey process linked with two vector valued functions satisfying a system of coupled non linear third order differential equations and of non linear heat equation.\\ \\
Let $N>1$ be an integer, $\a:=(a_1,...,a_N)$ with $a_1<...<a_N$ and $\k=(k_0,k_1,...,k_{N-1},k_N)$ such that $k_0=k_N=0$, $k_j\in[0,1]$, $j=1,...,N-1$ and $k_j\neq k_{j+1}, j=0,...,N-1$.\\
Consider the Pearcey process associated to the kernel $K_P$ \eqref{Kp} (which depends on an additional positive parameter $\tau$) and define its generating function 
\begin{equation}
\label{defF}
F(\a,\tau,\k):=\E\left(\prod_{j=1}^{N-1}(1-k_j)^{\sharp{(a_j,a_{j+1})}}\right)
\end{equation}
where $\sharp I$ is the random variable counting the number of points of the process in the interval $I$.\\
This generating function is well defined for $(a_i)_{1\leqslant i \leqslant N}$ finites.\\
Computing derivatives with respect to $k_j$'s of $F(\a,\tau,\k)$ allows us to express the joint probability of occupancy numbers of particles as follows:
\begin{equation}
\P\left(\bigcap_{j=1}^{N-1}\sharp(a_j,a_{j+1})=m_j\right)=\left.\dfrac{(-1)^{m_1+...+m_{N-1}}}{m_1!...m_{N-1}!}\dfrac{\partial^{m_1+...+m_{N-1}}}{\partial k_1^{m_1}...\partial k_{N-1}^{m_{N-1}}}F(\a,\tau,\k)\right._{\k=(1,...,1)}
\end{equation}
(see for instance \cite{Joh99} for a similar computation).
\begin{remark}
Since for all $B\subset\R$ bounded Borel set $\sharp B<\infty$ almost surely, $F(\a,\tau,\k)>0$.\\
Indeed, with Jensen's inequality,
$$F(\a,\tau,\k)\geqslant \exp\left(\sum_{j=1}^{N-1}\log(1-k_j)\E\left(\sharp(a_j,a_{j+1})\right)\right)$$
and $\forall 1\leqslant j\leqslant N-1,\ \log(1-k_j)\E\left(\sharp(a_j,a_{j+1})\right)>-\infty$.
\label{invK}
\end{remark}
Charlier and Moreillon studied the generating function of the Pearcey process \cite{CM21} on intervals of the form $(-x_j,x_j)$. They considered a parameter $r$ of dilatation of their intervals and obtained an expression for the generating function of the Pearcey process linked with a Hamiltonian. Recently, Kimura and Zahabi \cite{KZ22} introduced higher order Pearcey kernels and obtained a Hamiltonian structure for the level spacing distribution by studying a different Lax Pair for the Pearcey process on intervals of the form $[-s,s]$.\\ \\
In this article, we follow a different way. We consider the generating function \eqref{defF} not necessarily on symmetric intervals, and we use a parameter $s$ of translation of intervals instead of dilatation. This work is inspired by works of Claeys and Doeraene \cite{CD18} and Charlier and Doeraene \cite{CD17} where they studied the generating function for the Airy process and for the Bessel process. Recently Cafasso and Tarricone \cite{CT21} obtained a Tracy-Widom type formula for the generating function associated to the higher order Airy process.\\
The generating function associated to the Pearcey process satisfies the following:
\begin{theorem}
\label{mainthm}
Let $s\in\R$, $F(\a+s,\tau,\k)$ be as \eqref{defF} where $\a+s=(a_1+s,...,a_N+s)$, then
$$\dfrac{\partial^2}{\partial s^2}\log\left(F(\a+s,\tau,\k)\right)=p^T(s)q(s)$$
with $(p(s),q(s))=\left(p(s,\tau,\a,\k),q(s,\tau,\a,\k)\right)$ vectors of size $N$ satisfying:\\
$i$- a coupled third order differential equation\\
\begin{equation}
\left\{\begin{array}{l}
\partial_{sss}p^T+3(\partial_{s}p^T)qp^T+3p^Tq(\partial_{s}p^T)-\tau\partial_sp^T+p^TD_{s,\a}=0\\
\partial_{sss}q+3(\partial_{s}q)p^Tq+3qp^T(\partial_{s}q)-\tau\partial_sq-D_{s,\a}q=0
\end{array}\right.
\label{pqcoupledeq}
\end{equation}
where $D_{s,\a}:=\diag(a_1+s,...a_N+s)$ and\\
$ii$- a coupled non linear heat equation\\
\begin{equation}
\label{heatequation}
\left\{\begin{array}{l}
-\frac{1}{2}\partial_{ss}p^T-\partial_{\tau}p^T=p^Tqp^T\\
-\frac{1}{2}\partial_{ss}q+\partial_{\tau}q=qp^Tq
\end{array}\right. 
\end{equation}
Moreover $p_i(s)\sim \sqrt{k_i-k_{i-1}}P(a_i+s)$ and $q_i(s)\sim \sqrt{k_i-k_{i-1}}Q(a_i+s)$ as $s\to\infty$ where
\begin{equation}
\begin{array}{rlc}
Q(s)&\coloneqq \dfrac{1}{2i\pi}\displaystyle\int_{i\R}e^{-\frac{\mu^4}{4}+\tau\frac{\mu^2}{2}+s\mu}d\mu& \textit{and } Q'''(s)-\tau Q'(s)=sQ(s)\\
P(s)&\coloneqq \dfrac{1}{2i\pi}\displaystyle\int_{\Sigma}e^{\frac{\mu^4}{4}-\tau\frac{\mu^2}{2}-s\mu}d\mu& \textit{and } P'''(s)-\tau P'(s)=-sP(s)
\label{def:P et Q}
\end{array}
\end{equation}
\end{theorem}
\begin{remark}
Defining $u(s,\tau)\coloneqq\log\left(F(\a+s,\tau,\k)\right)$, from equations \eqref{pqcoupledeq} and \eqref{heatequation} we obtain (see Subsection \ref{appendix} in the Appendix for computations) that $u$ satisfies
\begin{equation}
\label{reducedKP}
\dfrac{\partial^2}{\partial \tau^2}u(s,\tau)+\dfrac{1}{2}\left(\dfrac{\partial^2}{\partial s^2}u(s,\tau)\right)^2+\dfrac{1}{12}\dfrac{\partial^4}{\partial s^4}u(s,\tau)-\dfrac{1}{3}\tau\dfrac{\partial^2}{\partial s^2}u(s,\tau)=0
\end{equation}
Equation \eqref{reducedKP} was already known for the gap probability for the Pearcey process (see for instance equation (1.10) in \cite{ACvM12}).
\end{remark}
The paper is organised as follows: in Section \ref{section2} we establish a link between $F(\a+s,\tau,\k)$ and a Fredholm determinant of an integrable operator in the sense of IIKS \cite{IIKS90} and introduce the Riemann-Hilbert problem associated to this integrable operator. In Section \ref{section3} we present a Lax pair associated to the Riemann-Hilbert problem for the Pearcey kernel and obtain a coupled vector differential equation with respect to the variable $s$ and a coupled non linear heat equation with respect to $s$ and $\tau$ for elements of the Riemann-Hilbert problem. Finally, in Section \ref{section4} we compute the logarithmic derivative of $F(\a+s,\tau,\k)$ with respect to $s$ and prove Theorem \ref{mainthm} using the results of Sections \ref{section2} and \ref{section3}.\\
\section{From generating function of the Pearcey process to a Riemann-Hilbert Problem}
\label{section2}
In this section we establish a link between the generating function $F(\a+s,\tau,\k)$ and the Fredholm determinant of an integrable operator and introduce the Riemann-Hilbert Problem (RHP) associated to this integrable operator.\\
\subsection*{The Pearcey kernel operator}
The Pearcey process is a determinantal point process on $\R$ associated to the Pearcey kernel operator.\\
For $(x,\tau)\in\R^2$ define $\theta_x(\mu):=\dfrac{\mu^4}{4}-\dfrac{\tau \mu^2}{2}-x\mu$ and introduce the Pearcey Kernel 
\begin{equation}
K_P(x,y;\tau):=\dfrac{1}{(2i\pi)^2}\int_{\Sigma}\int_{i\R}\dfrac{e^{\theta_x(\mu)-\theta_y(\lambda)}}{(\lambda-\mu)}d\lambda d\mu
\label{Kp}
\end{equation}
where $\Sigma:=\Sigma_-\cup\Sigma_+$ is as in the Figure \ref{fig:contour_opérateur}:

\begin{figure}[h]
    \begin{center}
        \includegraphics[scale = 1,trim = 8cm 21cm 8cm 4cm, clip]{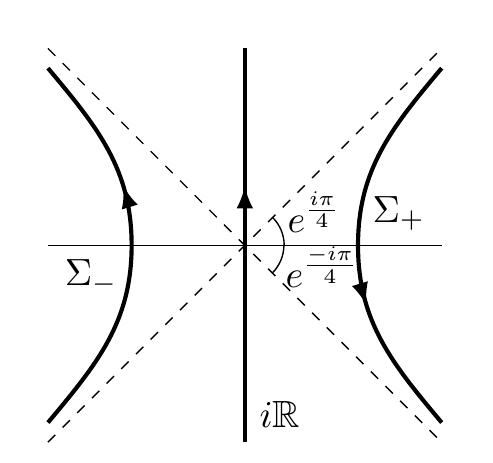}
        \caption{The contour for the operator $\K_P$}
        \label{fig:contour_opérateur}
    \end{center}
\end{figure}

Introduce $\K_P:L^2(\R)\to L^2(\R)$ with kernel $K_P$. According to Theorem 2 of \cite{Sos00}, if $s$ and $a_i$'s are all finite, then the generating function $F(\a+s,\tau,\k)$ and the Fredholm determinant of the operator $\K_P$ are linked in the following way:
\begin{equation}
F(\a+s,\tau,\k)=\det\left(1-\sum_{j=1}^{N-1}k_j\chi_{(a_j+s,a_{j+1}+s)}\K_P\right)
\label{FandKp}
\end{equation}
where $\chi_I$ is the characteristic function of the interval $I$.

\subsection*{From Pearcey kernel operator to an integrable operator}
We link the Fredholm determinant of the operator $\K_P$ with the one of an integrable operator in the sense of \cite{IIKS90} so that we study $F(\a+s,\tau,\k)$ with Riemann-Hilbert problem associated to integrable operator.\\
Define\begin{equation}
\tilde{f}(\mu)=\dfrac{1}{2i\pi}\left (\begin{array}{c}
e^{\frac{1}{2}\theta_0(\mu)}\chi_{\Sigma}(\mu) \\
\hline
\sqrt{k_1-k_0}e^{-\frac{1}{2}\theta_0(\mu)+(a_1+s)\mu}\chi_{i\R}(\mu)\\
\vdots\\
\sqrt{k_{N}-k_{N-1}}e^{-\frac{1}{2}\theta_0(\mu)+(a_N+s)\mu}\chi_{i\R}(\mu)
\end{array}\right)
\label{f}
\end{equation}
\begin{equation}
\label{g}
\tilde{g}(\mu)=\left (\begin{array}{c}
 e^{-\frac{1}{2}\theta_0(\mu)}\chi_{i\R}(\mu) \\
\hline
\sqrt{k_1-k_0}e^{\frac{1}{2}\theta_0(\mu)-(a_1+s)\mu}\chi_{\Sigma}(\mu)\\
\vdots\\
\sqrt{k_{N}-k_{N-1}}e^{\frac{1}{2}\theta_0(\mu)-(a_N+s)\mu}\chi_{\Sigma}(\mu)
\end{array}\right)
\end{equation}
We decompose $\tilde{f}$ and $\tilde{g}$ in a block $1\times 1$ and a vector of size $N$ because in what follows we will study a Riemann-Hilbert problem of size $(N+1)\times (N+1)$ and matrices will be partitioned in four blocks of size $1\times 1$, $1\times N$, $N\times 1$ and $N\times N$.\\
Define $\K:L^2(\Sigma\cup i\R)\to L^2(\Sigma\cup i\R)$ integrable operator in the sense of \cite{IIKS90} with kernel
\begin{equation}
K(u,v)=\dfrac{{\tilde{f}}^T(u)\tilde{g}(v)}{u-v}
\label{defK}
\end{equation}
We establish the relation between $F(\a+s,\tau,\k)$ and the Fredholm determinant of the operator $\K$ as follow:
\begin{proposition}
Let $F(\a+s,\k)$ be as in Theorem \ref{mainthm} and $\K$ be as in \eqref{defK}. Then
$$F(\a+s,\tau,\k)=\det(1-\K)_{L^2(\Sigma\cup i\R)}$$
\label{FandK}
\end{proposition}
\begin{proof}
The idea of the proof is to compose the operator $\K$ with multiplication operator and Fourier operator so that its Fredholm determinant will not change and to obtain an operator defined on $L^2(\R)$ with kernel $\displaystyle\sum_{j=1}^{N-1}k_j\chi_{(a_j+s,a_{j+1}+s)}(x)K_P(x,y)$.\\
The proof of this proposition is similar to the one of Theorem 4.1 in \cite{BC12} and can be adapted as follows. \\
Since $i\R$ and $\Sigma$ are disjoint, one can decompose $L^2(\Sigma\cup i\R)$ as $L^2(i\R)\oplus L^2(\Sigma)$ and write the following equality using matrix notation:
$$\det(1-\K)_{L^2(\Sigma\cup i\R)}=\det\left(1-\begin{pmatrix}
0 & \displaystyle\sum_{j=1}^N\mathcal{G}_j\\
\mathcal{F} & 0
\end{pmatrix}\right)_{L^2(i\R)\oplus L^2(\Sigma)}=\det\left(1-\sum_{j=1}^N\mathcal{G}_j\circ\mathcal{F}\right)_{L^2(i\R)}$$
where $\mathcal{F}$ and $\mathcal{G}_j$ are defined below with $\mathcal{G}_j$ depending on $k_j$'s
$$\begin{array}{rrrcl}
\mathcal{F}&:&L^2(i\R) & \longrightarrow & L^2(\Sigma)\\
 & &f &\longmapsto&\displaystyle\frac{e^{\frac{1}{2}\theta_0(\lambda)}}{2i\pi}\int_\Sigma\frac{e^{-\frac{1}{2}\theta_0(\mu)}}{\lambda-\mu}f(\mu)d\mu
\end{array}
$$
$$\begin{array}{rrrcl}
\mathcal{G}_j&:&L^2(\Sigma) & \longrightarrow & L^2(i\R)\\
 & &g &\longmapsto&\displaystyle(k_j-k_{j-1})\frac{e^{-\frac{1}{2}\theta_0(\xi)+\xi(a_j+s)}}{2i\pi}\int_\Sigma\frac{e^{\frac{1}{2}\theta_0(\lambda)-\lambda(a_j+s)}}{\xi-\lambda}g(\lambda)d\lambda
\end{array}
$$
Then composing with a multiplication operator $\mathcal{M}$ and using a Fourier operator $\mathcal{T}$(for example as in \cite{Gir14}) it is possible to relate Fredholm determinants of $\K$ and $\K_P$. Here the Fourier composition allows to go from an operator on $L^2(i\R)$ to one on $L^2(\R)$.\\
With $\mathcal{M}:=e^{-\frac{1}{2}\theta_0(\mu)}$ and
$$\begin{array}{rrrcl}
\mathcal{T}&:&L^2(i\R) & \longrightarrow & L^2(\R)\\
 & &f &\longmapsto&\displaystyle\dfrac{1}{\sqrt{2i\pi}}\int_{i\R}e^{-\xi x}f(\xi)d\xi
\end{array}$$
$\mathcal{T}\circ\mathcal{M}^{-1}\circ\mathcal{G}_j\circ\mathcal{F}\circ\mathcal{M}\circ\mathcal{T}^{-1}$ has kernel
$$\mathcal{L}_j(x,y)=\dfrac{k_j-k_{j-1}}{(2i\pi)^2}\int_{i\R}e^{\xi(a_j+s-x)}\int_{i\R}\int_\Sigma\dfrac{e^{\theta_{a_j+s}(\lambda)-\theta_y(\mu)}}{(\xi-\lambda)(\lambda-\mu)}d\lambda d\mu \dfrac{d\xi}{2i\pi}$$
$$\mathcal{L}_j(x,y)=\left\{\begin{aligned}
\dfrac{k_j-k_{j-1}}{(2i\pi)^2}\int_{i\R}\int_{\Sigma_+}\dfrac{e^{\theta_x(\lambda)-\theta_y(\mu)}}{(\mu-\lambda)}d\lambda d\mu,&\ x>a_j+s\\
-\dfrac{k_j-k_{j-1}}{(2i\pi)^2}\int_{i\R}\int_{\Sigma_-}\dfrac{e^{\theta_x(\lambda)-\theta_y(\mu)}}{(\mu-\lambda)}d\lambda d\mu, &\ x<a_j+s\end{aligned}\right.
$$
If $x<a_1+s$, then $x<a_j+s$ for all $j$ and $\sum_{j=1}^N\mathcal{T}\circ\mathcal{M}^{-1}\circ\mathcal{G}_j\circ\mathcal{F}\circ\mathcal{M}\circ\mathcal{T}^{-1}$ has kernel
$$-\dfrac{\sum_{j=1}^Nk_j-k_{j-1}}{(2i\pi)^2}\int_{\Sigma_-}\int_{i\R}\dfrac{e^{\theta_x(\mu)-\theta_y(\lambda)}}{(\lambda-\mu)}d\lambda d\mu=-\dfrac{k_N-k_{0}}{(2i\pi)^2}\int_{\Sigma_-}\int_{i\R}\dfrac{e^{\theta_x(\mu)-\theta_y(\lambda)}}{(\lambda-\mu)}d\lambda d\mu=0$$
The same holds for $x>a_N+s$.\\
If $x\in(a_j+s,a_{j+1}+s)$, then \begin{align*}
\sum_{j=1}^N\mathcal{L}_j(x,y)&=\dfrac{\sum_{\ell=1}^jk_\ell-k_{\ell-1}}{(2i\pi)^2}\int_{i\R}\int_{\Sigma_+}\dfrac{e^{\theta_x(\lambda)-\theta_y(\mu)}}{(\mu-\lambda)}d\lambda d\mu\\&-\dfrac{\sum_{\ell=j+1}^Nk_\ell-k_{\ell-1}}{(2i\pi)^2}\int_{i\R}\int_{\Sigma_-}\dfrac{e^{\theta_x(\lambda)-\theta_y(\mu)}}{(\mu-\lambda)}d\lambda d\mu\\
&=k_j\chi_{(a_j+s,a_{j+1}+s)}(x)K_P(x,y)
\end{align*}
This concludes the proof.
\end{proof}
We can study $\det(1-\K)_{L^2(\Sigma\cup i\R)}$ with the theory of Riemann-Hilbert problems associated to an integrable operator. We describe the Riemann-Hilbert problem associated to $\K$ in what follows.\\
\subsection*{Riemann-Hilbert Problem associated to Pearcey kernel operator}
The contours for the Riemann-Hilbert problem associated to $F(\a+s,\tau,\k)$ is $\Sigma\cup i\R$ oriented as in the previous figure.\\
For the jump matrix, introduce $$f(\mu):=\begin{pmatrix}
-\sqrt{k_1-k_0}e^{-\theta_{a_1+s}(\mu)}\\
\vdots\\
-\sqrt{k_{N}-k_{N-1}}e^{-\theta_{a_N+s}(\mu)}
\end{pmatrix}\chi_{i\R}(\mu) \text{ and}$$ $$g(\mu):=\begin{pmatrix}
-\sqrt{k_1-k_0}e^{\theta_{a_1+s}(\mu)}\\
\vdots\\
-\sqrt{k_{N}-k_{N-1}}e^{\theta_{a_N+s}(\mu)}
\end{pmatrix}\chi_{\Sigma}(\mu)$$

\begin{RHP}[RHP for $\Gamma$]
\label{RHP}
We consider the Riemann-Hilbert problem with contours $\Sigma\cup i\R$ and jump matrix $J(\mu;\a,\tau,s):=\left (\begin{array}{c|c}
1 & g^T(\mu) \\
\hline
f(\mu) & I_n 
\end{array}\right)=I_{N+1}-2i\pi\tilde{f}(\mu)\tilde{g}^T(\mu)$ where $I_k\in M_k(\C)$ is the identity matrix.\\We search a matrix valued function $\Gamma(\mu)=\Gamma(\mu;\a,\tau,s)$ such that:
\begin{itemize}
\item $\Gamma:\C\setminus\Sigma\cup i\R\to\mathcal{G}\ell_{N+1}(\C)$ is analytic
\item $\Gamma_+(\mu)=\Gamma_-(\mu)J(\mu)$, $\mu\in\Sigma\cup i\R$ where $\Gamma$ is continuous up to boundary of the contours and $\Gamma_\pm(\mu)$ are non-tangential limits approaching $\mu$ from left(+) or right(-).
\item 
\begin{equation}
\Gamma(\mu)=I_{N+1}+\displaystyle\sum_{j\geqslant 1}\dfrac{\Gamma_j}{\mu^j}=I_{N+1}+\dfrac{1}{\mu}\left (\begin{array}{c|c}
-\delta(\a,\tau,s) & p^T(\a,\tau,s) \\
\hline
q(\a,\tau,s) & \Delta(\a,\tau,s)
\end{array}\right)+\mathcal{O}\left(\mu^{-2}\right)\text{ as } \mu\to\infty
\label{asymp}
\end{equation}
\end{itemize}
\end{RHP}
\begin{remark}
If $\Gamma$ satisfies the previous RHP, then $\det(\Gamma)$ is entire and according to the asymptotic $\det(\Gamma)\equiv 1$.
Since $\det(\Gamma)\equiv 1$, the previous RHP has an unique solution (if it exists), $\Tr(\Gamma_1)=0$ and $\delta=\Tr(\Delta)$.\\
\label{det1tr0}
\end{remark}
\section{Study of the Riemann-Hilbert Problem associated to the Pearcey kernel operator}
\label{section3}
In this section, we obtain a Lax Pair associated to the previous RHP which leads to a system of coupled vector differential equation for $p^T$ and $q$ (see equation \eqref{asymp} for the definition of $p^T$ and $q$). Studying the derivative with respect to $\tau$ yields to a coupled non linear heat equation satisfied by $p^T$ and $q$.

\subsection*{A Lax Pair for $\Psi$}
Let $\Gamma$ be a solution of $RHP$ \ref{RHP} and $T$ be the gauge transformation:
$$T(\mu):= e^{\dfrac{1}{N+1}\displaystyle\sum_{j=1}^N\theta_{a_j+s}(\mu)} \diag\left(1,e^{-\theta_{a_1+s}(\mu)},...,e^{-\theta_{a_N+s}(\mu)}\right).$$
Defining $\Psi(\mu):=\Gamma(\mu)T(\mu)$, the following result holds for $\Psi$.
\begin{proposition}[Lax pair for $\Psi$]
If $\Gamma$ is solution of RHP \ref{RHP} then $\Psi$ satisfy a system of partial differential equation.
\begin{equation}
\left\{\begin{array}{c}
\partial_s\Psi(\mu)=A(\mu)\Psi(\mu)\\
\partial_\mu\Psi(\mu)=B(\mu)\Psi(\mu)\\
\partial_\tau\Psi(\mu)=C(\mu)\Psi(\mu)
\end{array}\right.\\
\label{Laxpair}
\end{equation}
where \begin{equation}
A(\mu)=\mu A_1+A_0=\dfrac{\mu}{N+1}\left (\begin{array}{c|c}
-N & (0) \\
\hline
(0) & I_N 
\end{array}\right)+\left (\begin{array}{c|c}
0 & p^T \\
\hline
-q & (0) 
\end{array}\right)
\end{equation}
\begin{equation}
B(\mu)=\mu^3\tilde{B}_3+\mu^2B_2+\mu(\tilde{B}_1+B_1)+\tilde{B}_0+B_0
\end{equation} with $B_j$'s depending on $\Gamma_j$'s and 
$\tilde{B}_3=-A_1,\ \tilde{B}_1=\tau A_1$ $$\tilde{B}_0=sA_1+\dfrac{1}{N+1}\diag\left(-\sum_{j=1}^Na_j,Na_1-\sum_{j\neq 1}a_j,...,Na_N-\sum_{j\neq N}a_j\right)$$
and
\begin{equation}
C(\mu)=\mu^2\tilde{C}_2+\mu C_1+C_0=\dfrac{\mu^2}{2}A_1+\mu C_1+C_0
\end{equation}
with $C_1$ and $C_0$ depending on $\Gamma_j$'s.
\end{proposition}
\begin{proof}
It is easy to show that if $\Gamma$ satisfies the previous $RHP$ then $\Psi$ satisfies a RHP with jump matrix which does not depend on $\a,\tau,s$ and $\mu$ (the jump only depends on \\$k_j, j=0,...,N$).

Then $\Psi$ and $\partial_{s}\Psi$ have the same jump on the contours. From this fact we deduce that
$A(\mu):=\partial_{s}\Psi(\mu)\Psi(\mu)^{-1}$ is entire. Using Liouville's theorem with asymptotics of $\Psi$ and $\partial_{s}\Psi$ from their RHP we conclude $A$ is a polynomial of degree 1 in $\mu$. More precisely, we compute $A$ and obtain:
$$A(\mu)=\mu A_1+A_0=\dfrac{\mu}{N+1}\left (\begin{array}{c|c}
-N & (0) \\
\hline
(0) & I_N 
\end{array}\right)+\left (\begin{array}{c|c}
0 & p^T \\
\hline
-q & (0) 
\end{array}\right)$$

Using the same method we conclude that $B(\mu):=\partial_{\mu}\Psi(\mu)\Psi(\mu)^{-1}$ (respectively $C(\mu):=\partial_{\tau}\Psi(\mu)\Psi(\mu)^{-1}$) is a polynomial of degree 3 (respectively degree 2).
We will not do all computations for $B$ and $C$ with the asymptotic as we did for $A$: we precise what we will do.\\
Computations with the asymptotic involve $T$,  its partial derivative with respect to $\mu$ (respectively $\tau$) and $(\Gamma_j)_{j\geqslant 1}$. For now we will only compute terms in $B$ (respectively $C$) which does not depend on $(\Gamma_j)_{j\geqslant 1}$.\\
We start with $B$ and write it as:
$$B(\mu)=\mu^3\tilde{B}_3+\mu^2{B_2}+\mu(\tilde{B}_1+B_1)+\tilde{B}_0+B_0$$ where only $(B_j)_{j\in\{0,1,2\}}$ (respectively $(\tilde{B}_j)_{j\in\{0,1,3\}}$) depends (respectively does not depend) on $(\Gamma_j)_{j\geqslant 1}$.
$$B(\mu)=(\Gamma_\mu T+\Gamma T_\mu)T^{-1}\Gamma^{-1}=\Gamma_\mu\Gamma^{-1}+\Gamma T_\mu T^{-1}\Gamma^{-1}$$
$\Gamma_\mu\Gamma^{-1}=\mathcal{O}(\mu^{-2})$ as $\mu\to\infty$, then
$$B(\mu)\sim\left(I_{N+1}+\dfrac{\Gamma_1}{\mu}+\mathcal{O}\left(\mu^{-2}\right)\right)T_\mu T^{-1}\left(I_{N+1}-\dfrac{\Gamma_1}{\mu}+\mathcal{O}\left(\mu^{-2}\right)\right)$$
$B(\mu)=T_\mu T^{-1}+\displaystyle\sum_{j=0}^2\mu^jB_j$ because of Liouville's theorem. We compute $T_\mu T^{-1}$ and obtain:
$$\tilde{B}_3=-A_1,\ \tilde{B}_1=\tau A_1$$ $$\tilde{B}_0=sA_1+\dfrac{1}{N+1}\diag\left(-\sum_{j=1}^Na_j,Na_1-\sum_{j\neq 1}a_j,...,Na_N-\sum_{j\neq N}a_j\right)$$
Similarly $C(\mu)=T_\tau T^{-1}+\displaystyle\sum_{j=0}^1\mu^j C_j$ and $T_\tau T^{-1}=\dfrac{\mu^2}{2}A_1$.\\
\end{proof}
\begin{remark}
\label{deltas}
For $\delta$, $p^T$ and $q$ as in \eqref{asymp} the following holds: $\partial_s\delta=-p^Tq$.\\
Actually, computing the term in $1/\mu$ in the asymptotic of $A$, we obtain $\partial_s\Gamma_1+[\Gamma_2,A_1]-[\Gamma_1,A_1]\Gamma_1$. Then, because of Liouville's theorem this term is $0$ and the block $1\times 1$ on the diagonal block matrix leads to the equation. This equation will be useful later when we will compute the logarithmic derivative for $F(\a+s,\tau,\k)$.\\
\end{remark}
\begin{proposition}
Let $p^T$ and $q$ be as in \eqref{asymp}. Then they satisfy the following coupled vector $3^{rd}$ order differential equation and non linear coupled heat equation:

\begin{equation}
\left\{\begin{array}{l}
\partial_{sss}p^T+3(\partial_{s}p^T)qp^T+3p^Tq(\partial_{s}p^T)-\tau\partial_sp^T+p^TD_{s,\a}=0\\
\partial_{sss}q+3(\partial_{s}q)p^Tq+3qp^T(\partial_{s}q)-\tau\partial_sq-D_{s,\a}q=0
\end{array}\right.
\end{equation}

\begin{equation}
\left\{\begin{array}{l}
-\frac{1}{2}\partial_{ss}p^T-\partial_{\tau}p^T=p^Tqp^T\\
-\frac{1}{2}\partial_{ss}q+\partial_{\tau}q=qp^Tq
\end{array}\right. .
\end{equation}
\label{prop8}
\end{proposition}
\begin{proof}
The compatibility condition for the Lax pair of $\Psi$ leads to the equation
\begin{equation}
\partial_sB-\partial_\mu A=[A,B]
\label{compcond}
\end{equation}
We use the same approach as in \cite{WE07} and \cite{CT21}. If we write $B_j^{ki}$ blocs of matrix $B_j$ where $B_j^{11}$ is a scalar, $B_j^{12}$ is a row of size $N$, $B_j^{21}$ a column of size $N$ and $B_j^{22}$ a $N\times N$ matrix. Then \eqref{compcond} gives a polynomial equation in $\mu$ and we obtain an equation for every monomial. This leads to the following equations:
\begin{equation}
B_2^{12}=-p^T,\ B_2^{21}=q
\label{pq}
\end{equation}
\begin{equation}
\left\{\begin{array}{l}
\partial_sB_j^{11}=p^TB_j^{21}+B_j^{12}q\\
\partial_sB_j^{12}=-B_{j-1}^{12}\delta_{j\neq 0}+p^TB_j^{22}-B_j^{11}p^T+\tau p^T\delta_{j,1}+p^TD_{s,\a}\delta_{j,0}\\
\partial_sB_j^{21}=B_{j-1}^{12}\delta_{j\neq 0}-qB_j^{11}+B_j^{22}q+\tau q\delta_{j,1}+D_{s,\a}q\delta_{j,0}\\
\partial_sB_j^{22}=-qB_j^{12}-B_j^{21}p^T\\
\end{array}\right.
\label{systeme_eq_Bj}
\end{equation}
where $D_{s,\a}=\diag(a_1+s,...,a_N+s)$ and $\delta_{i,j}$ is the Kronecker delta.\\
We define formally the operator $\partial_s^{-1}$ such that $\partial_s^{-1}\partial_s=1$. From the first and the last equation of \eqref{systeme_eq_Bj} we obtain 
$B_j^{11}=\partial_s^{-1}\left(p^TB_j^{21}+B_j^{12}q\right)$ and $B_j^{22}=-\partial_s^{-1}\left(qB_j^{12}+B_j^{21}p^T\right)$.
With $j=2$, $$B_2^{11}=\partial_s^{-1}(0)=c_2^{11},\ \ B_2^{22}=\partial_s^{-1}((0))=c_2^{22}$$ with $c_2^{11}$ and $c_2^{22}$ independent of $s$. Actually with the asymptotics we obtain:
$B_2=-A_0$. Then $c_0^{11}=0$ and $c_0^{22}=(0)$.
Same for $B_1$, even if it depends on $\Gamma_2$, diagonal terms only depend on $\Gamma_1$. The asymptotic leads to
$$B_1^{11}=p^Tq,\ \ B_1^{22}=-qp^T$$
Using second and third equations of \eqref{systeme_eq_Bj} we compute $B_1^{12}$ and $B_1^{21}$ with $j=2$ then $B_0^{12}$ and $B_0^{21}$ with $j=1$.
$$B_1^{12}=\partial_sp^T,\ \ B_1^{21}=\partial_sq$$
$$B_0^{12}=-\partial_{ss}p^T-2p^Tqp^T+\tau p^T,\ \ B_0^{21}=\partial_{ss}q+2qp^Tq-\tau q$$
Using first and last equations of \eqref{systeme_eq_Bj} with $j=0$ and integrating, we compute $B_0^{11}$ and $B_0^{22}$.
$$B_0^{11}=p^T(\partial_{s}q)-(\partial_{s}p^T)q$$
$$B_0^{22}=q(\partial_{s}p^T)-(\partial_{s}q)p^T$$
Finally, it remains two equations (second and third of \eqref{systeme_eq_Bj} with $j=0$). Replacing $B_0^{ki}$ in these equations we obtain a system of equations satisfied by $p^T$ and $q$:
\begin{equation}
\label{chazyequation}
\left\{\begin{array}{l}
\partial_{sss}p^T+3(\partial_{s}p^T)qp^T+3p^Tq(\partial_{s}p^T)-\tau\partial_sp^T+p^TD_{s,\a}=0\\
\partial_{sss}q+3(\partial_{s}q)p^Tq+3qp^T(\partial_{s}q)-\tau\partial_sq-D_{s,\a}q=0
\end{array}\right.
\end{equation}
Studying the compatibility condition of $A$ and $C$ the same way as for $A$ and $B$, $p^T$ and $q$ satisfied a coupled non-linear heat equation:
\begin{equation}
\label{heat}
\left\{\begin{array}{l}
-\frac{1}{2}\partial_{ss}p^T-\partial_{\tau}p^T=p^Tqp^T\\
-\frac{1}{2}\partial_{ss}q+\partial_{\tau}q=qp^Tq
\end{array}\right. .
\end{equation}
\end{proof}
In appendix A of \cite{KZ22}, doing formal computation on a semi infinite interval, Kimura and Zahabi obtained a scalar version of the system of coupled differential equation \eqref{chazyequation}.\\
Similar equations as in Proposition \ref{prop8} appeared in the study of the limiting one-point distribution of periodic TASEP \cite{BLS22}. The authors obtained coupled mKdV equations and coupled non linear heat equations. Combining these two equations, they proved the second log-derivative of the Fredholm determinant they studied satisfies the second Kadomtsev-Petviashvili equation. It is possible to do a similar computation and to obtain a PDE for the second log-derivative of the Fredholm determinant of $\K$. The next section will be partially devoted to this computation.
\section{The logarithmic derivative of \textit{F} and proof of Theorem \ref{mainthm}}
\label{section4}
Finally we prove Theorem \ref{mainthm}.\\
From Remark \ref{invK} and Proposition \ref{FandK}, the Fredholm determinant of $(1-\K)$ is different from zero and $1-\K$ is invertible.\\
Defining
\begin{equation}
\tilde{F}:=(1-\K)^{-1}\tilde{f},
\label{Ftilde},\ \ \text{with }\tilde{f} \text{ as in equation } \eqref{f},
\end{equation}
if we write $\tilde{F}=\left (\begin{array}{c}
F_0 \\
\hline
F_1\\
\vdots\\
F_N
\end{array}\right)$ and $\tilde{g}=\left (\begin{array}{c}
g_0 \\
\hline
g_1\\
\vdots\\
g_N
\end{array}\right)$ (where $\tilde{g}$ is defined in equation \eqref{g}), we have the following result.
\begin{lemma}
Let $\Delta$ be as in \eqref{asymp} and $F_i$'s, $g_i$'s as above. Then:
$$\Delta=\int_{\Sigma\cup i\R}\left (\begin{array}{c}
F_1(\mu)\\
\vdots\\
F_N(\mu)
\end{array}\right)\left(g_1(\mu),\dots,g_N(\mu)\right)d\mu$$
\label{lemmaDelta}
\end{lemma}
\begin{proof}
According to theory of Riemann-Hilbert problem, since $(1-\K)$ is invertible, the resolvent of $\K$ and the unique solution to RHP \ref{RHP} are linked and $\tilde{F}=\Gamma_+\tilde{f}$.\\
With this last equality we obtain:
$$\Gamma(\xi)=I_{N+1}-\int_{\Sigma\cup i\R}\dfrac{\tilde{F}(\mu)\tilde{g}^T(\mu)}{\mu-\xi}d\mu$$ 
Expanding $\dfrac{1}{\mu-\xi}$ we express $\Gamma_1$ in function of $\tilde{F}$ and $\tilde{g}$.
$$\Gamma_1=\int_{\Sigma\cup i\R}\tilde{F}(\mu)\tilde{g}^T(\mu)d\mu$$ According to this previous equality and because of the decomposition by blocks of $\Gamma_1$ in \eqref{RHP}, with \begin{equation}
\tilde{F}=\left (\begin{array}{c}
F_0 \\
\hline
F_1\\
\vdots\\
F_N
\end{array}\right) \text{ and } \tilde{g}=\left (\begin{array}{c}
g_0 \\
\hline
g_1\\
\vdots\\
g_N
\end{array}\right),
\label{Fg}
\end{equation} $$\Delta=\int_{\Sigma\cup i\R}\left (\begin{array}{c}
F_1(\mu)\\
\vdots\\
F_N(\mu)
\end{array}\right)\left(g_1(\mu),\dots,g_N(\mu)\right)d\mu.$$
\end{proof}
\begin{proposition}
\label{prop:logFdelta}
Let $F(\a+s,\tau,\k)$ be as in Theorem \ref{mainthm} and $\delta$ as in \eqref{asymp}. The following holds:
$$\dfrac{\partial}{\partial s}\log(F(\a+s,\tau,\k))=-\delta$$
\end{proposition}
\begin{proof}
$$\dfrac{\partial}{\partial s}\log(F(\a+s,\tau,\k))=\dfrac{\partial}{\partial s}\log(\det(1-\K))=\dfrac{\partial}{\partial s}\Tr(\log(1-\K))=-\Tr\left((1-\K)^{-1}\partial_s\K\right)$$
Let $(e_n)_{n\in\N}$ be an orthonormal basis of $L^2(\Sigma\cup i\R)$.\\
$$\Tr((1-\K)^{-1}\partial_s\K)=\sum_{n\in\N}\langle(1-\K)^{-1}\partial_s\K e_n,e_n\rangle$$
But $\partial_s\K$ has kernel $\chi_{i\R}(u)\tilde{f}^T(u)\tilde{g}(v)\chi_\Sigma(v)$.
Then $$\Tr((1-\K)^{-1}\partial_s\K)=\sum_{n\in\N}\langle(1-\K)^{-1}(\chi_{i\R}\tilde{f}^T),e_n\rangle\langle\chi_{\Sigma}\tilde{g},e_n\rangle$$
According to \eqref{Fg}, definitions of $\tilde{f}$ \eqref{f}, $\tilde{g}$ \eqref{g} and $\tilde{F}$ \eqref{Ftilde},
$$\Tr((1-\K)^{-1}\partial_s\K)=\sum_{n\in\N}\left(\langle F_1,e_n\rangle,\dots,\langle F_N,e_n\rangle\right)\left (\begin{array}{c}
\langle g_1,e_n\rangle\\
\vdots\\
\langle g_N,e_n\rangle
\end{array}\right)=\Tr(\Delta)$$
The last equation is a consequence of Lemma \ref{lemmaDelta} and the fact that $e_n$ is an orthonormal basis.
Finally, because of $-\delta+\Tr(\Delta)=\Tr(\Gamma_1)=0$ (see Remark \ref{det1tr0}),
$$\dfrac{\partial}{\partial s}\log(F(\a+s,\tau,\k))=-\delta$$
\end{proof}
We use the previous proposition and the discussion on the Lax Pair to prove Theorem \ref{mainthm}.
\begin{proof}[Proof of Theorem \ref{mainthm}]
Using the previous proposition we derive $\dfrac{\partial}{\partial s}\log\left(F(\a+s,\tau,\k)\right)$ with respect to $s$. 
\begin{equation}
\dfrac{\partial^2}{\partial s^2}\log\left(F(\a+s,\tau,\k)\right)=-\partial_s\delta=p^T(s)q(s)
\end{equation}
because $\partial_s\delta=-p^Tq$ (see Remark \ref{deltas}). From the Proposition \ref{prop8} $p^T$ and $q$ satisfy equations \eqref{pqcoupledeq} and \eqref{heatequation}.\\

It remains to compute the asymptotics for $p_i(s)$ and $q_i(s)$ as $s\to\infty$.\\
Introduce the function $X$ defined by $X(\mu)\coloneqq \Gamma(s^{1/3}\mu)$. $X(\mu)$ has jumps on $\Sigma\cup i\R$ of the form $\left (\begin{array}{c|c}
1 & g^T(s^{1/3}\mu) \\
\hline
f(s^{1/3}\mu) & I_n 
\end{array}\right)$ where $$f(s^{1/3}\mu)=\begin{pmatrix}
-\sqrt{k_1-k_0}e^{-s^{4/3}\left(\frac{\mu^4}{4}-\frac{\tau\mu}{2s^{2/3}}-(1+\frac{a_1}{s})\mu\right)}\\
\vdots\\
-\sqrt{k_{N}-k_{N-1}}e^{-s^{4/3}\left(\frac{\mu^4}{4}-\frac{\tau\mu}{2s^{2/3}}-(1+\frac{a_N}{s})\mu\right)}
\end{pmatrix}\chi_{i\R}(\mu) \text{ and}$$ $$g(s^{1/3}\mu)=\begin{pmatrix}
-\sqrt{k_1-k_0}e^{s^{4/3}\left(\frac{\mu^4}{4}-\frac{\tau\mu}{2s^{2/3}}-(1+\frac{a_1}{s})\mu\right)}\\
\vdots\\
-\sqrt{k_{N}-k_{N-1}}e^{s^{4/3}\left(\frac{\mu^4}{4}-\frac{\tau\mu}{2s^{2/3}}-(1+\frac{a_N}{s})\mu\right)}
\end{pmatrix}\chi_{\Sigma}(\mu)$$

According to Figure \ref{fig:phase small norm}, one can deform the contour $\Sigma\cup i\R$ into $\tilde{\Sigma}\cup i\R$ so that the jump on $i\R$ and $\Sigma_+$ tends to the identity matrix $I_{N+1}$ as $s\to\infty$. Such a deformation for the jump on the contour $\Sigma_-$ is not possible.
\begin{figure}[!h]
    \begin{center}
        \includegraphics[scale = 1,trim = 10cm 4cm 10cm 4cm, clip]{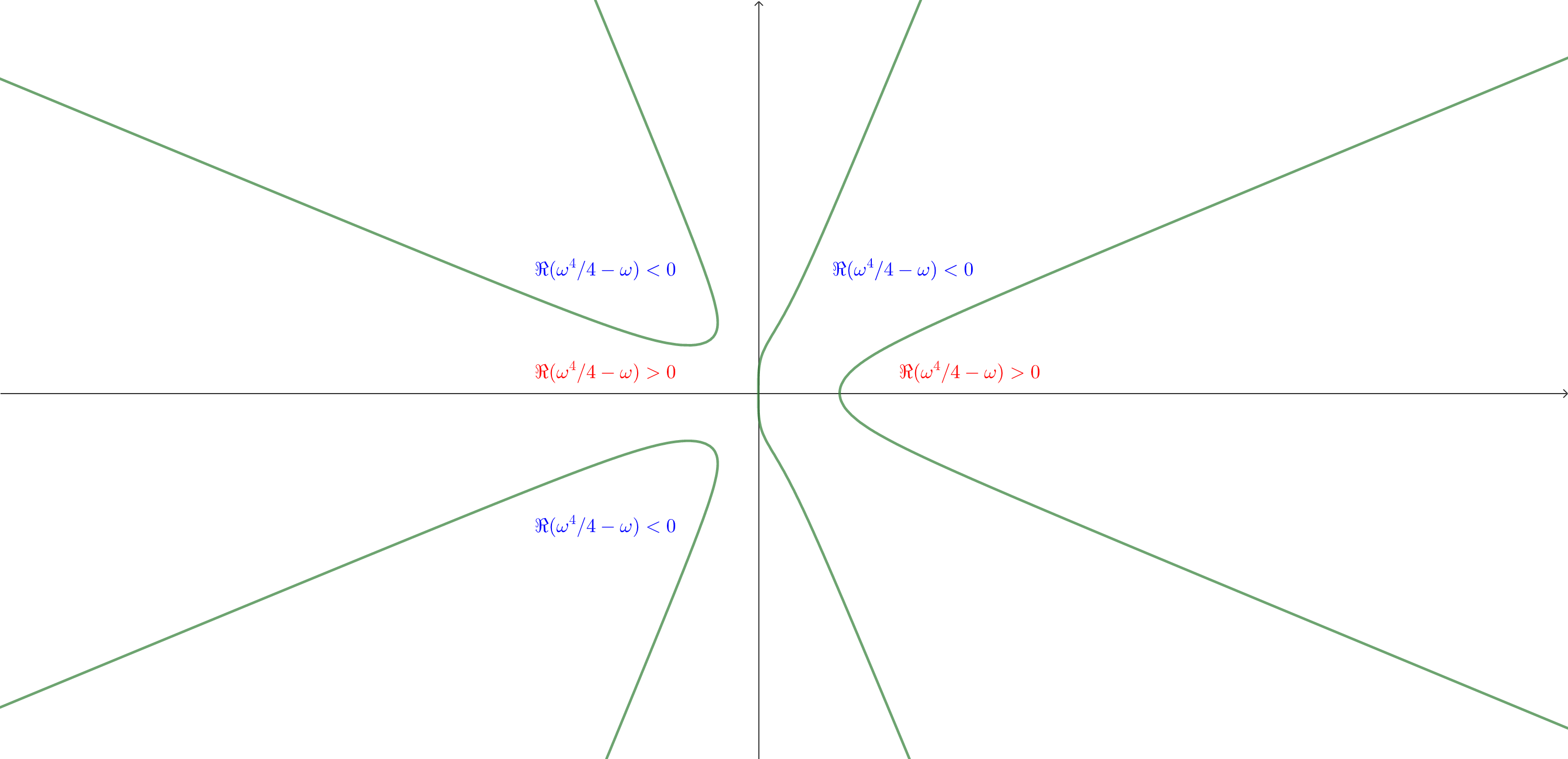}
        \caption{The sign of $\Re(\omega^4/4-\omega)$}
        \label{fig:phase small norm}
    \end{center}
\end{figure}

Then $X$ satisfies the following Riemann-Hilbert problem.
\begin{RHP}[RHP for $X$]
\label{RHPX}
We search a matrix valued function $X(\mu)=X(\mu;\a,\tau,s)$ such that:
\begin{itemize}
\item $X:\C\setminus\tilde{\Sigma}\cup i\R\to\mathcal{G}\ell_{N+1}(\C)$ is analytic
\item $X_+(\mu)=X_-(\mu)J(s^{1/3}\mu)$, $\mu\in\tilde{\Sigma}\cup i\R$ where $X$ is continuous up to boundary of the contours and $X_\pm(\mu)$ are non-tangential limits approaching $\mu$ from left(+) or right(-).
\item 
\begin{equation}
X(\mu)=I_{N+1}+\displaystyle\sum_{j\geqslant 1}\dfrac{X_j}{\mu^j}\text{ as } \mu\to\infty
\label{eq:asympX}
\end{equation}
\end{itemize}
\end{RHP}
We want to study the asymptotic of $X$ as $s\to\infty$ by applying the small norm theorem (see for instance Theorem 5.1.5 in \cite{Its11}).
\begin{lemma}
\label{lemma:limit X}
Let $X$ satisfies the Riemann-Hilbert Problem \ref{RHPX}. Then, as $s\to\infty$,
\begin{equation}
X(\mu,s)\to I_{N+1}+T(\mu) 
\end{equation} 
where $T$ is a upper triangular matrix with zero diagonal of the form $T(\mu)=\left (\begin{array}{c|c}
0 & t(\mu) \\
\hline
0_N & 0_{N\times N} 
\end{array}\right)$ with $0_N$ and $0_{N\times N}$ the zero matrix of size $N\times 1$ and $N\times N$ and $t$ a matrix of size $1\times N$.
\end{lemma}
\begin{proof}
We introduce the lower and upper triangular matrices with zero diagonal $$T_-(\mu)=\tilde{f}(\mu)\tilde{g}^T(\mu)\chi_{i\R}(\mu) \text{ and } T_+(\mu)=\tilde{f}(\mu)\tilde{g}^T(\mu)\chi_{\Sigma}(\mu)$$ and the function $\tilde{X}$ solution of the following Riemann-Hilbert problem:
\begin{itemize}
\item $\tilde{X}:\C\setminus\tilde{\Sigma}\to\mathcal{G}\ell_{N+1}(\C)$ is analytic
\item $\tilde{X}_+(\mu)=\tilde{X}_-(\mu)\left(I_{N+1}-2i\pi T_+(s^{1/3}\mu)\right)$, $\mu\in\tilde{\Sigma}$ where $\tilde{X}$ is continuous up to boundary of the contours and $\tilde{X}_\pm(\mu)$ are non-tangential limits approaching $\mu$ from left(+) or right(-).
\item 
$$\tilde{X}(\mu)=I_{N+1}+\mathcal{O}(\mu^{-1})\text{ as } \mu\to\infty$$
\end{itemize}
Define $\mathcal{X}(\mu)\coloneqq X(\mu)\tilde{X}(\mu)^{-1}$, the function $\mathcal{X}$ satisfies a Riemann-Hilbert problem on the contours $i\R$ with the jump matrix $I_{N+1}-2i\pi T_-(s^{1/3}\mu)$. The matrix $T_-(s^{1/3}\mu)$ satisfies: let $C\coloneqq \sup_{i\in\{1,...,N\}}\frac{\mid\sqrt{k_i-k_{i-1}}\mid}{2\pi}$ then $$\parallel T_-(s^{1/3}\mu)\parallel_{L^\infty(i\R)}=\sup_{\mu\in i\R}Ce^{-s^{4/3}\left(\frac{\mu^4}{4}-\frac{\tau\mu^2}{2s^{2/3}}-\left(1+\frac{a_i}{s}\right)\mu\right)}\to 0 \text{ and }$$ $\parallel T_-\left(s^{1/3}\mu\right)\parallel_{L^2(i\R)}\to 0$
exponentially fast as $s\to\infty$. We conclude from the small norm theorem that $\mathcal{X}\sim I_{N+1}$ as $s\to\infty$ which implies $X\sim I_{N+1}+T$ as $s\to\infty$ where $T$ is a upper triangular matrix with zero diagonal of the form $T(\mu)=\left (\begin{array}{c|c}
0 & t(\mu) \\
\hline
0_N & 0_{N\times N} 
\end{array}\right)$ with $0_N$ and $0_{N\times N}$ the zero matrix of size $N\times 1$ and $N\times N$ and $t$ a matrix of size $1\times N$.\\
\end{proof}
From Lemma \ref{lemma:limit X} we deduce asymptotics for $p_i$ and $q_i$ and we finish the proof of Theorem \ref{mainthm}.\\
As for $\Gamma_1$ in the proof of Lemma \ref{lemmaDelta}, one can express $X_1(s)$ as $\displaystyle\int_{\tilde{\Sigma}\cup i\R}X_-(\mu)\tilde{f}(s^{1/3}\mu)\tilde{g}^T(s^{1/3}\mu)d\mu$.\\
Moreover $\Gamma_1(s)$ and $X_1(s)$ satisfy the following relation, $X_1(s)=s^{-1/3}\Gamma_1(s)$. Then
\begin{equation}
\begin{array}{rl}q_i(s)&=\Gamma_1(s)_{i+1,1}=s^{1/3}X_1(s)_{i+1,1}=\left(s^{1/3}\displaystyle\int_{\tilde{\Sigma}\cup i\R}X_-(\mu)\tilde{f}(s^{1/3}\mu)\tilde{g}^T(s^{1/3}\mu)d\mu\right)_{i+1,1}\\
 &\sim\dfrac{\sqrt{k_i-k_{i-1}}}{2i\pi}\displaystyle\int_{i\R}e^{-\frac{\mu^4}{4}+\tau\frac{\mu^2}{2}+(a_i+s)\mu}d\mu=\sqrt{k_i-k_{i-1}}Q(a_i+s)\end{array}
\label{def:Q}
\end{equation} where $Q$ satisfies the differential equation $Q'''(s)-\tau Q'(s)=sQ(s)$.\\
\begin{equation}
\begin{array}{rl}p_i(s)&=\Gamma_1(s)_{1,i+1}=s^{1/3}X_1(s)_{1,i+1}=\left(s^{1/3}\displaystyle\int_{\tilde{\Sigma}\cup i\R}X_-(\mu)\tilde{f}(s^{1/3}\mu)\tilde{g}^T(s^{1/3}\mu)d\mu\right)_{1,i+1}\\
 &\sim\dfrac{\sqrt{k_i-k_{i-1}}}{2i\pi}\displaystyle\int_{\Sigma}e^{\frac{\mu^4}{4}-\tau\frac{\mu^2}{2}-(a_i+s)\mu}d\mu=\sqrt{k_i-k_{i-1}}P(a_i+s)\end{array}
 \label{def:P}
\end{equation} where $P$ satisfies the differential equation $P'''(s)-\tau P'(s)=-sP(s)$.
\end{proof}

Asymptotics for $P$ and $Q$ as $s\to\infty$ are computed in the appendix.
\section*{Acknoledgements}
I am very grateful to Mattia Cafasso for many advices given and his support during the preparation of this paper. This work has been supported by the European Union Horizon 2020 research and innovation program under the Marie Sklodowska-Curie RISE 2017 grant agreement no. 778010 IPaDEGAN and by the IRP Probabilités Intégrables, Intégrabilité Classique et Quantique (PIICQ), funded by the CNRS. 

\section{Appendix}
\subsection{Asymptotics of $P$ and $Q$}
As $s\to\infty$, $P$ and $Q$ have the following asymptotic.
\begin{proposition}
Let $P$ and $Q$ defined respectively as in equations \eqref{def:P} and \eqref{def:Q}. Then, as $s\to\infty$:
\begin{equation}
Q(s)\displaystyle\sim_{s\to\infty}\sqrt{\dfrac{2}{3\pi}}s^{-1/3}e^{-\frac{3}{8}s^{4/3}-\frac{\tau}{4}s^{2/3}+\frac{\tau^2}{6}}\cos\left(\frac{3}{4}\sin\left(\frac{2\pi}{3}\right)s^{4/3}-\frac{\tau}{2}\sin\left(\frac{2\pi}{3}\right)s^{2/3}-\frac{\pi}{6}\right)
\end{equation}
\begin{equation}
P(s)\displaystyle\sim_{s\to\infty}\sqrt{\dfrac{2}{3\pi}}s^{-1/3}e^{\frac{3}{8}s^{4/3}+\frac{\tau}{4}s^{2/3}-\frac{\tau^2}{6}}\cos\left(\frac{3}{4}\sin\left(\frac{2\pi}{3}\right)s^{4/3}-\frac{\tau}{2}\sin\left(\frac{2\pi}{3}\right)s^{2/3}-\frac{\pi}{6}\right)
\end{equation}
\end{proposition}
\begin{proof}
By definition,
$$Q(s)=\dfrac{1}{2i\pi}\displaystyle\int_{i\R}e^{-\frac{\mu^4}{4}+\tau\frac{\mu^2}{2}+s\mu}d\mu$$
The derivative with respect to $\mu$ of $\theta_s(\mu)$ is $\frac{d}{ds}\theta_s(\mu)=\mu^3-\tau\mu-s$ whose roots are 
\begin{equation}
\mu_k\coloneqq\mu_k(s,\tau)=j^k\sqrt[3]{\frac{1}{2}\left(s+\sqrt{s^2-\frac{2^2\tau^3}{3^3}}\right)}+j^{-k}\sqrt[3]{\frac{1}{2}\left(s-\sqrt{s^2-\frac{2^2\tau^3}{3^3}}\right)}
\end{equation} where $j:=e^{\frac{2\pi}{3}}$ and $k\in\{0,1,2\}$. We only consider saddle points for $k=1$ and $2$.\\
Denote by $\mu_*$ either $\mu_1$ or $\mu_2$. One can deform the contour $i\R$ into a contour $\mathcal{C}_Q$ so that it passes through $\mu_1$ and $\mu_2$ and the following holds $\Re\left(\theta_s(\mu)-\theta_s(\mu_*)\right)<0$.\\
See Figure \ref{fig:Rethetamu-thetamu*} for the study of the sign of $\Re\left(\theta_s(\mu)-\theta_s(\mu_*)\right)$ for $s=100$ and $\tau=1$. As $s\to\infty$ and $\tau$ is fixed, the algebraic curve $\Re\left(\theta_s(\mu)-\theta_s(\mu_*)\right)=0$ keeps a similar shape and it is always possible to deform $i\R$ into $\mathcal{C}_Q$.
\begin{figure}[h]
    \begin{center}
        \includegraphics[scale = 0.4,trim = 4cm 1cm 10cm 2cm, clip]{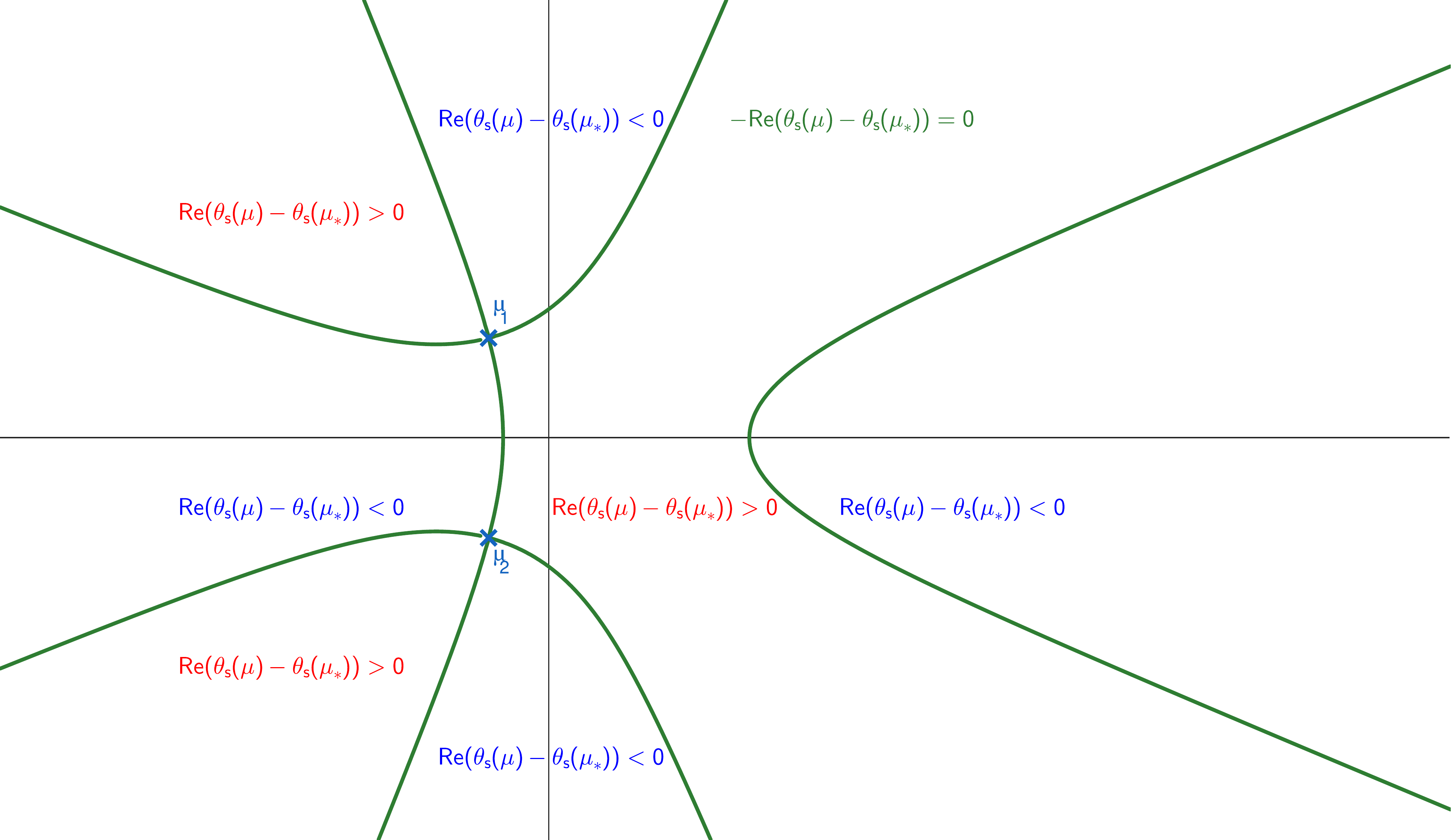}
        \caption{The sign of $\Re(\theta_s(\mu)-\theta_s(\mu_*))$ for $\tau=1$ and $s=100$}
        \label{fig:Rethetamu-thetamu*}
    \end{center}
\end{figure}

Then \begin{equation}
\begin{array}{rl}
Q(s)&=\dfrac{e^{-\theta_s(\mu_*)}}{2i\pi}\displaystyle\int_{\mathcal{C}_Q}e^{-\left(\theta_s(\mu)-\theta_s(\mu_*)\right)}d\mu\\
&=\dfrac{e^{-\theta_s(\mu_*)}}{2i\pi}\displaystyle\int_{\mathcal{C}_Q}e^{-\left(\frac{3}{2}\mu_*^2(\mu-\mu_*)^2\left(1-\frac{\tau}{3\mu_*^2}+\frac{2}{3\mu_*}(\mu-\mu_*)+\frac{1}{6\mu_*^2}(\mu-\mu_*)^2\right)\right)}d\mu
\end{array}
\end{equation}
and one can approximate $Q$ as $s\to\infty$ by expanding $\mu_*$ and by approximating the integral.
It follows
\begin{equation}
Q(s)\sim\sqrt{\dfrac{2}{3\pi}}s^{-1/3}e^{-\frac{3}{8}s^{4/3}-\frac{\tau}{4}s^{2/3}+\frac{\tau^2}{6}}\cos\left(\frac{3}{4}\sin\left(\frac{2\pi}{3}\right)s^{4/3}-\frac{\tau}{2}\sin\left(\frac{2\pi}{3}\right)s^{2/3}-\frac{\pi}{6}\right)
\end{equation}
The same method leads to the asymptotic for $P$ as $s\to\infty$.
\end{proof}

\subsection{Computations for equation \eqref{reducedKP}}
\label{appendix}
We now prove equation \eqref{reducedKP}.\\
Define $u(s,\tau)\coloneqq \log\left(F(\a+s,\tau,\k)\right)$ and $v(s,\tau)\coloneqq \dfrac{\partial^2}{\partial s^2}u(s,\tau)$. According to Proposition \ref{prop:logFdelta}, $v(s,\tau)=p^T(s)q(s)$. Differentiating $v$ with respect to $\tau$ and using equation \eqref{heat} to express $\partial_\tau p^T$ and $\partial_\tau q$, we obtain:
\begin{equation}
\partial_\tau v=\dfrac{1}{2}\partial_s\left(p^T\partial_sq-\partial_sp^Tq\right)
\end{equation}
Differentiating a second time $v$ with respect to $\tau$ (again using equation \eqref{heat}) yields to
\begin{equation}
\partial_{\tau\tau}v=\dfrac{1}{2}\left(2v\partial_sv+\dfrac{1}{2}\left(p^T\partial_{sss}q+\partial_{sss}p^Tq\right)-\dfrac{1}{2}\partial_s\left(\partial_sp^T\partial_sq\right)\right)
\label{etape1}
\end{equation}
Recall $v(s,\tau)=p^T(s)q(s)$, then differentiating three times with respect to $s$ the following equation holds:
\begin{equation}
\partial_s\left(\partial_sp^T\partial_sq\right)=\dfrac{1}{3}\left(\partial_{sss}v-\left(p^T\partial_{sss}q+\partial_{sss}p^Tq\right)\right)
\end{equation}
Replacing $\partial_s\left(\partial_sp^T\partial_sq\right)$ in equation \eqref{etape1} and using equation \eqref{chazyequation} we obtain:
\begin{equation}
\partial_{\tau\tau}v=\partial_s\left(-v\partial_sv-\dfrac{1}{12}\partial_{sss}v+\dfrac{1}{3}\tau\partial_sv\right).
\end{equation}
Replacing $v$ by $\partial_{ss}u$ and integrating twice with respect to $s$ we prove $u$ satisfies equation \eqref{reducedKP}.

\end{document}